\documentclass[envcountset]{llncs}
\usepackage{amsmath}
\usepackage{amssymb}
\usepackage{graphicx}

\usepackage{mathtools}
\usepackage{tikz}
\usetikzlibrary{arrows}
\usepackage{caption}
\usepackage{float}
\usepackage{comment}

\newcommand{\sub}{\subseteq}
\newcommand{\tuple}[1]{\vec{#1}}
\newcommand {\key}[1] {k({#1})}

\newcommand {\indepc}[2] {#1 ~\bot~ #2}

\newcommand{\Dom}{\textrm{Dom}}

\newcommand{\Si}{\Sigma}
\newcommand{\on}{\exists}
\newcommand{\ja}{\wedge}

\newcommand{\yli}{\overline}

\newcommand{\Cu}[1]{\textrm{Cl}_{\rm \uparrow}(#1)}

\def\boto{\mkern1.5mu\bot\mkern2.5mu}
\newcommand{\vdashh}{\vdash_{\mathfrak{I}}}
\newcommand{\ax}[1]{\mathcal{R}{#1}}

\begin{document}

\author{Miika Hannula\inst{1} \and Juha Kontinen\inst{1}  \and Sebastian Link\inst{2} }

\institute{University of Helsinki, Department of Mathematics and Statistics, Helsinki, Finland \email{\{miika.hannula,juha.kontinen\}@helsinki.fi} \and {University of Auckland, Department of Computer Science, New Zealand \texttt{s.link@auckland.ac.nz}}}

\title{On Independence Atoms and  Keys\thanks{The first two authors were supported by grant 264917 of the Academy of Finland.}}

\maketitle

\begin{abstract}
Uniqueness and independence are two fundamental properties of data. Their enforcement in database systems can lead to higher quality data, faster data service response time, better data-driven decision making and knowledge discovery from data. The applications can be effectively unlocked by providing efficient solutions to the underlying implication problems of keys and independence atoms. Indeed, for the sole class of keys and the sole class of independence atoms the associated finite and general implication problems coincide and enjoy simple axiomatizations. However, the situation changes drastically when keys and independence atoms are combined. We show that the finite and the general implication problems are already different for keys and unary independence atoms. Furthermore, we establish a finite axiomatization for the general implication problem, and show that the finite implication problem does not enjoy a $k$-ary axiomatization for any $k$.
\end{abstract}

\section{Introduction}
Keys and independence atoms are two classes of data dependencies that enforce the uniqueness and independence of data in database systems. Keys are one of the most important classes of integrity constraints as effective data processing largely depends on the identification of data records. Their importance is manifested in the de-facto industry standard for data management, SQL, and they enjoy native support in every real-world database system. A relation $r$ satisfies the key $k(X)$ for a set $X$ of attributes, if for all tuples $t_1,t_2\in r$ it is true that $t_1=t_2$ whenever $t_1$ and $t_2$ have matching values on all the attributes in $X$. Independence atoms also occur naturally in data processing, including query languages. For example, one of the most fundamental operators in relational algebra is the Cartesian product, combining every tuple from one relation with every tuple from a second relation. In SQL, users must specify this database operation in form of the \texttt{FROM} clause. A relation $r$ satisfies the independence atom $X\boto Y$ between two sets $X$ and $Y$ of attributes, if for all tuples $t_1,t_2\in r$ there is some tuple $t\in r$ which matches the values of $t_1$ on all attributes in $X$ and matches the values of $t_2$ on all attributes in $Y$. In other words, in relations that satisfy $X\boto Y$, the occurrence of $X$-values is independent of the occurrence of $Y$-values. Due to their fundamental importance in everyday data processing in practice, both keys and independence atoms have also received detailed interest from the research community since the 1970s \cite{DBLP:journals/jacm/BojanczykMSS09,DBLP:journals/ipl/Demetrovics78,DBLP:journals/tcs/DemetrovicsKMST98,DBLP:journals/tods/Fagin81,geiger:1991,DBLP:conf/wollic/KontinenLV13,DBLP:journals/jcss/LucchesiO78,DBLP:conf/icdt/NiewerthS11,paredaens:1980,DBLP:journals/is/Wijsen09}. One of the core problems studied for approximately 100 different classes of relational data dependencies alone are their associated implication problems \cite{thalheim:1991}. Efficient solutions to these problems have their applications in database design, query and update processing, data cleaning, exchange, integration and security to name a few. Both classes of keys and independence atoms in isolation enjoy efficient computational properties: finite and general implication problems coincide, and are axiomatizable by finite sets of Horn rules, respectively \cite{paredaens:1980,geiger:1991,DBLP:conf/wollic/KontinenLV13,thalheim:1991}.

Given their importance for data processing in practice, given that keys and independence atoms naturally co-exist and given the long and fruitful history of research into relational data dependencies, it is rather surprising that keys and independence atoms have not been studied together. For an illustrative example of their interaction consider the SQL query $Q$
\begin{center}
\begin{tabular}{l@{\hspace*{1.25cm}}l}
Query $Q$: & Query $Q'$: \\
\texttt{SELECT} p.id, \texttt{COUNT}(\texttt{DISTINCT} s.id) & \texttt{SELECT} p.id, \texttt{COUNT}(s.id) \\
\texttt{FROM} \textsc{part} p, \textsc{supplier} s           & \texttt{FROM} \textsc{part} p, \textsc{supplier} s\\
\texttt{GROUP BY} p.id                                       & \texttt{GROUP BY} p.id \\
\end{tabular}
\end{center}
which returns for each part (identified by p.id) the number of distinct possible suppliers (identified by s.id). Here, the command \texttt{DISTINCT} is used to eliminate duplicate suppliers. In data processing duplicate elimination is time-consuming and not executed by default. However, duplicate elimination in query $Q$ is redundant. The \texttt{GROUP BY} clause uses p.id values to partition the Cartesian product $\textsc{part}\times\textsc{supplier}$, generated by the \texttt{FROM} clause, into sub-relations. That is, each sub-relation satisfies the independence atom $\textsc{part}\boto\textsc{part}$. As the Cartesian product $\textsc{part}\times\textsc{supplier}$ satisfies the key $k(p.id,s.id)$, so does each of its sub-relations. However, the key $k(p.id,s.id)$ and the independence atom $\textsc{part}\boto\textsc{part}$ together imply the key $k(s.id)$. Hence, there are no duplicate s.id values in any sub-relation and $Q$ can be replaced by the more efficient query $Q'$.

Motivated by these strong real-world applications and the lack of previous research we study the interaction of key dependencies and independence atoms. Somewhat surprisingly, the good computational properties that hold for each class in isolation do not carry over to the combined class. In fact, we show that for the combined class of keys and independence atoms:
\begin{itemize}
\item The finite and the general implication problem differ from one another,
\item For keys and unary independence atoms the general implication problem has a 2-ary axiomatization by Horn rules, but
\item Their finite implication problem is not finitely axiomatizable.
\end{itemize}

Our results are somewhat similar to those known for the combined class of functional dependencies (FDs) and inclusion dependencies (INDs). While both classes in isolation have matching finite and general implication problems and enjoy finite axiomatizations, the finite and the general implication problem differ for the combined class of FDs and unary INDs already \cite{DBLP:journals/jcss/CasanovaFP84}. For FDs and unary INDs the general implication problem has a 2-ary axiomatization by Horn rules \cite{DBLP:journals/jacm/CosmadakisKV90}, while their finite implication problem is not finitely axiomatizable \cite{DBLP:journals/jcss/CasanovaFP84}. Interestingly, key dependencies are strictly subsumed by FDs. It is also known that both implication problems are undecidable for FDs and INDs \cite{chandra85,mitchell83}, but decidable for FDs and unary INDs \cite{DBLP:journals/jacm/CosmadakisKV90}. We would also like to mention that independence atoms form an efficient fragment of embedded multivalued dependencies whose expressivity results in the non-axiomatizability of its implication problem by a finite set of Horn rules \cite{parker:1980} and its undecidability \cite{herrmann:2006}.

Our work is further motivated by the recent development of the area of dependence logic constituting a novel approach to the study of various notions of dependence and independence that is intimately linked with databases and their data dependencies \cite{gradel10,vaananen07}. It has been shown recently, e.g. that the general implication problem of so-called conditional independence atoms and inclusion atoms can be finitely axiomatized in this context \cite{DBLP:conf/foiks/HannulaK14}. For databases, this result establishes a finite axiomatization (utilizing implicit existential quantification) of the general implication problem for inclusion, functional, and  embedded multivalued dependencies taken together. This result is similar to the axiomatization of the general implication problem for FDs and INDs \cite{mitchell83}.


\section{Preliminaries}

\subsection{Definitions}
A \emph{relation schema} $R$ is a set of symbols $A$ called \emph{attributes}, each equipped with a domain $\Dom(A)$ representing the possible values that can occur in the column named $A$. A tuple $t$ over $R$ is a mapping $R \rightarrow \bigcup_{A \in R} \Dom(A)$ where $t(A) \in \Dom(A)$ for each $A \in R$. For a tuple $t$ over $R$ and $R' \sub R$, $t(R')$ is the restriction of $t$ on $R'$.  A \emph{relation} $r$ over $R$ is a set of tuples $t$ over $R$. If $R' \sub R$ and $r$ is a relation over $R$, then we write $r(R')$ for $\{t( R') : t\in r\}$. If $A\in R$ is an attribute and $r$ is a relation over $R $,  then we write $r(A = a)$ for $\{t \in r: t(A)=a\}$. For sets of attributes $X$ and $Y$, we often write $XY$ for $X \cup Y$, and denote singleton sets of attributes $\{A\}$ by $A$. Also, for a relation schema $A_1\ldots A_n$, a relation $r(A_1\ldots A_n)$ is sometimes identified with the set notation $\{(a_1, \ldots ,a_n) \mid \on t\in r : t(A_i)= a_i$ for $1 \leq i \leq n\}$.


\subsection{Independence Atoms and Key Dependencies}

Let $R$ be a relation schema and $X \sub R$. Then $\key{X}$ is a $R$-\emph{key}, given the following semantic rule for a relation $r$ over $R$:
\begin{itemize}
\item[S-K]$ r \models \key{X}$ if and only if $\forall t,t' \in r: t(X) = t'(X) \Rightarrow t=t'$.
\end{itemize}

Let $R$ be a relation schema and $X,Y \sub R$. Then $X \boto Y$ is a $R$-\emph{independence atom}, given the following semantic rule for a relation $r$ over $R$:
\begin{itemize}
\item[S-I]$ r \models X \boto Y$ if and only if $\forall t,t' \in r\on t''\in r: t''(X)=t(X) \ja t''(Y) = t'(Y)$.
\end{itemize}
An independence atom $\indepc{X}{Y}$ is called unary if $X$ and $Y$ are single attributes. $R$-keys and $R$-independence atoms are together called $R$-\emph{constraints}. If $\Si$ is a set of $R$-constraints and $R' \sub R$, then we write $\Si \upharpoonright R'$ for the subset of all $R'$-constraints of $\Si$.

\subsection{Implication Problems}

For a set $\Sigma\cup\{\phi\}$ of independence atoms and keys we say that $\Sigma$ \textit{implies} $\phi$, written $\Sigma\models\phi$, if every relation that satisfies every element in $\Sigma$ also satisfies $\phi$. We write  $\Sigma\models_{\rm FIN}\phi$, if  every finite relation that satisfies every element in $\Sigma$ also satisfies $\phi$. We say that $\phi$ is a \emph{$k$-ary} (finite) \emph{implication} of $\Si$, if there exists $\Si' \sub \Si$ such that $|\Si'| \leq k$ and $\Si' \models \phi$ ($\Si' \models_{\rm FIN} \phi$).

In this article we consider the axiomatizability of the so-called finite and the general  implication problem for unary independence atoms and keys. The general implication problem for independence atoms and keys  is defined as follows.
\begin{center}
\begin{tabular}{|@{\hspace*{.25cm}}l@{\hspace*{.1cm}}l@{\hspace*{.25cm}}|}\hline
PROBLEM: & General implication problem for independence atoms and keys\\ \hline
INPUT:   & Relation schema $R$, \\
         & Set $\Sigma\cup\{\varphi\}$ of independence atoms and keys  over $R$ \\
OUTPUT:  & Yes, if $\Sigma\models\varphi$; No, otherwise \\ \hline
\end{tabular}
\end{center}
The finite implication problem is defined analogously by replacing  $\Sigma\models \phi$ with  $\Sigma\models _ {\rm FIN} \phi$.

For a set $\mathfrak{I}$ of inference rules,  we denote by $\Sigma\vdash_\mathfrak{I}\phi$  the \emph{inference} of $\phi$ from $\Sigma$. That is, there is some sequence $\gamma=[\sigma_1,\ldots,\sigma_n]$ of independence atoms and keys such that $\sigma_n=\phi$ and every $\sigma_i$ is an element of $\Sigma$ or results from an application of an inference rule in $\mathfrak{I}$ to some elements in $\{\sigma_1,\ldots,\sigma_{i-1}\}$.
 A set $\mathfrak{I}$ of inference rules is said to be \emph{sound} for the general implication problem of independence atoms and keys, if for every $R$ and for every set $\Sigma$,
 $\Sigma\vdash_\mathfrak{I}\phi$ implies that $\Sigma\models \phi$.
A set $\mathfrak{I}$ is called \emph{complete} for the general implication problem if
 $\Sigma\models \phi$ implies that $\Sigma\vdash_\mathfrak{I}\phi$.
The (finite) set $\mathfrak{R}$ is said to be a (finite) \emph{axiomatization} of the general implication for independence atoms and keys if $\mathfrak{R}$ is both sound and complete. These notions are defined analogously for the finite implication problem. For $k\geq 1$, an axiomatization $\mathfrak{R}$ is called $k$-ary if all the rules of $\mathcal{R}$ are of the form
$$
\cfrac{A_1 \quad A_2 \quad \ldots \quad  A_{l-1} \quad A_l}{B}
$$
where $l \leq k$.

\section{The General Implication Problem}
In this section we shown that the below set of axioms $\mathfrak{I}$ is complete for the general implication problem of unary independence atoms and arbitrary keys taken together.

\begin{table}[h]
\[\fbox{$\begin{array}{c@{\hspace*{.25cm}}c@{\hspace*{.25cm}}c}
\cfrac{}{\emptyset\boto X} & \cfrac{X\boto Y}{Y\boto X} & \cfrac{X \boto X \quad Y \boto Z }{XY \boto Z} \\
\text{(trivial independence, $\mathcal{R}1$)} & \text{(symmetry, $\mathcal{R}2$)} & \text{(constancy, $\mathcal{R}3$)}

\\ \\

\cfrac{X\boto YZ}{X\boto Y} & \cfrac{X\boto Y\quad XY\boto Z}{X\boto YZ}  &  \cfrac{}{\key{R}}   \\
 \text{(decomposition, $\mathcal{R}4$)} & \text{(exchange, $\mathcal{R}5$)}  &  \text{(trivial key, $\mathcal{R}6$)}

\\ \\

\cfrac{\key{X}}{\key{XY}} &     \cfrac{X \boto X \quad \key{XY}}{\key{Y}} & \cfrac{X\boto Y\quad \key{X}}{Y\boto Y}  \\
\text{(upward closure, $\mathcal{R}7$)}& \text{(1st composition, $\mathcal{R}8$)} & \text{(2nd composition, $\mathcal{R}9$)}

\end{array}$}\]
\caption{Axiomatization $\mathfrak{I}$ of Independence Atoms and Keys in Database Relations\label{tab-rules}}
\end{table}

It is straightforward to check the soundness of the  axioms $\mathfrak{I}$.

\begin{theorem}
The axioms $\mathfrak{I}$ are sound for the general implication problem of independence atoms and keys.
\end{theorem}
Next we will show that the set of axioms $\mathfrak{I}$ is complete for the general implication problem of unary independence atoms and arbitrary keys.
\begin{theorem}\label{partial1}
Assume that $R$ is a relation schema and $\Sigma \cup\{\phi\}$ consists of $R$-keys and unary $R$-independence atoms. Then  $\Sigma \vdashh \phi$ iff $\Sigma \models \phi$.
\end{theorem}

\begin{proof}
Assume to the contrary that $\Si \not\vdashh \phi$. We will construct a countably infinite relation witnessing $\Si \not\models \phi$. 
Let $\Si_{\rm i} \cup \Si_{\rm k}$ be the partition of $\Si$ to independence atoms and keys, respectively. Let  $\indepc{X_1}{Y_1}, \ldots ,\indepc{X_{N}}{Y_{N}}$ be an enumeration of $\Si_{\rm i}$, and let $A_1, \ldots ,A_M$ be an enumeration of $R$. Moreover, let $R':=\{A \in R : \Si \vdashh \indepc{A}{A}\}$. We will construct an increasing chain (with respect to $\sub$) of finite relations $r_n$, for $n \geq 0$, such that
\begin{enumerate}
\item $r_n(R') = \{\tuple 0\}$,
\item $r_n \models \Si_{\rm k}$, and $r_n \models \indepc{X_l}{Y_l}$ if $1 \leq n= l$ modulo $N$.
\end{enumerate}
Then letting $r:= \bigcup_{n \geq 0} r_n$, we obtain that $r \models \Si$. Regarding $\phi$, we also have two cases:  $\phi$ is either of the form \begin{itemize}
\item[$(i)$] $\key{D} $ or
\item[$(ii)$]$X \boto Y$.
\end{itemize}
For showing that $r\not\models \phi$, it suffices to define relations  $r_n$ so that $r_0:=\{t_0,t_1\}$ where
\begin{itemize}
\item[$3.$]  $\begin{cases}r_0 \not\models \key{D}&\textrm{ in case }(i), \\
  \textrm{ for no } t\in r_n:t(XY)=t_0(X)t_1(Y)&\textrm{ in case }(ii).
\end{cases}$
\end{itemize}
Relations $r_n$ are now constructed inductively as follows:
\begin{itemize}
\item Assume first that $n=0$. We let $r_0:=\{t_0,t_1\}$ where, for $1 \leq i \leq M$,
\begin{itemize}
\item $t_0(A_i):=$ $\begin{cases}
0 &\textrm{ if }A_i \in R'D \textrm{ in case $(i)$ or }A_i \in R'  \textrm{ in case }(ii),\\
i & \textrm{ otherwise},
\end{cases}$
\item $t_1(A_i):=$ $\begin{cases}
0 &\textrm{ if }A_i \in R'D \textrm{ in case $(i)$ or }A_i \in R'  \textrm{ in case }(ii),\\
M+i & \textrm{ otherwise}.
\end{cases}$
\end{itemize}
Then item 1 follows from the definition. For showing that $r_0 \models \Si_{\rm k}$, let $\key{B} \in \Si_{\rm k}$. Assume to the contrary that $r_0 \not\models \key{B}$. Then we have two cases:
\begin{itemize}

\item In case $(i)$, $B \sub R'D$ when we  obtain that $\Si \vdashh B \cap R' \boto B \cap R'$  using repeatedly $\mathcal{R}2$ and $\mathcal{R}3$. From this, since $\key{B} \in \Si$, we obtain $\key{B \setminus R'}$ with $\mathcal{R}8$. Since $B \setminus R' \sub D$, we then obtain $\phi$ with $\mathcal{R}7$. This again contradicts with the assumption $\Si \not\vdashh \phi$.

\item In case $(ii)$, $B \sub R'$, when we obtain that $\Si \vdashh B \boto X$ using first $\ax{1}$ and then repeatedly $\mathcal{R}3$. From this, since $\key{B}\in \Si$, we then obtain $X \boto X$ by $\mathcal{R}9$. From $X \boto X$ we  obtain $\phi$ with $\ax{1}$ and $\mathcal{R}3$ which contradicts with the assumption $\Si \not\vdashh \phi$.

\end{itemize}
Hence $r_0 \models \key{B}$ when we obtain that $r_0 \models  \Si_{\rm k}$.
For item 3, note that in case $(i)$, $r_0 \not\models \key{D}$ by the definition of $r_0$. Also in case $(ii)$ where $\phi$ is $X\boto Y$,  we must have $XY \sub R\setminus R'$, since otherwise we would obtain that $\Si \vdashh \phi$ using $\mathcal{R}3$ and $\mathcal{R}2$. Thus by the definition of $r_0$, we conclude that for no $t \in r_0:t(XY)=t_0(X)t_1(Y)$.

\item Assume then that $r_{n}$ is a finite relation satisfying the induction assumption; we will construct a finite relation $r_{n+1}$ also satisfying the induction assumption. Assume that $l = n+1$ modulo $N$. If $r_{n} \models X_l \boto Y_l$, then we let $r_{n+1}:=r_{n}$. Otherwise, let $(a_1,b_1), \ldots ,(a_k,b_k)$ be an enumeration of $r_{n}(X_l)\times r_n(Y_l) \setminus r_{n}(X_lY_l)$, and assume that $m$ is the maximal number occurring in $r_n$. We then let $r_{n+1}$ be obtained by extending $r_n$ with tuples $s_i$, for $1 \leq i \leq k$, such that $s_i(A_j)$, for $1 \leq j \leq M$, is defined as follows:
\begin{itemize}
\item $s_i(A_j)=$ $\begin{cases}
0 &\textrm{ if }A_j \in R',\\
a_i &\textrm{ if } A_j=X_l,\\
b_i &\textrm{ if } A_j =Y_l,\\
m + iM +j & \textrm{ otherwise}.
\end{cases}$

\end{itemize}
Note that $r_{n+1}$ is well defined: from the assumption $r_n \not\models X_l \boto Y_l$ and the induction assumption $r_n(R')=\{\yli{0}\}$ we obtain that
\begin{equation}\label{ei}
X_l,Y_l \not\in R'
\end{equation}
from which it also follows that $X_l$ and $Y_l$ are two distinct attributes.
Now item 1 of the claim and $r_{n+1} \models X_l\boto Y_l$ of item 2 follow from the definition. For showing that $r_{n+1} \models \Si_{\rm k}$, let $\key{B} \in \Si_{\rm k}$. Assume to the contrary that $r_{n+1} \not\models \key{B}$. Then, by the definition of $r_{n+1}$, and since $r_n \models \key{B}$ by the induction assumption, we obtain that $B \sub R'X_l$ or $B \sub R'Y_l$.  Assume first that $B \sub R'X_l$. Since $\key{B} \in \Si$, we then obtain $\key{R'X_l}$ by $\mathcal{R}7$. By the definition of $R'$, we obtain $R' \boto R'$ using repeatedly $\mathcal{R}2$ and $\ax{3}$. Then from $R' \boto R'$ and $\key{R'X_l}$, we obtain $\key{X_l}$ by $\ax{8}$. From this and $X_l \boto Y_l$ we would then, by $\ax{9}$, obtain $Y_l \boto Y_l$ when $Y_l\in R'$ contradicting with \eqref{ei}. The case where $B \sub R'Y_l$ is analogous. Therefore the counter-assumption $r_{n+1} \not\models \key{B}$ is false, and hence $r_{n+1} \models \Si_{\rm k}$.

For item 3 of the claim, assume that $\phi$ is $X \boto Y$. Assume to the contrary that for some $t \in r_{n+1}\setminus r_n:t(XY)=t_0(X)t_1(Y)$. First recall that $XY \sub R\setminus R'$ because $\Si \not\vdashh \phi$ when by the definition of $r_{n+1}$, we obtain that $XY \sub X_lY_l$. Moreover, by the assumption and  the definition of $t_0$ and $t_1$, it follows that $X$ and $Y$ are two distinct attributes. Hence $X \boto Y$  is either $X_l \boto Y_l$ or $Y_l\boto X_l$. Since $X_l \boto Y_l \in \Si$, we then,  by $\ax{2}$, obtain that $\Si\vdashh \phi$ which contradicts with the assumption. Hence item 3 of the induction assumption also holds. This concludes the construction of the relations $r_{n}$.

\end{itemize}
By the above construction, taking $r:= \bigcup_{n \geq 0} r_n$, we obtain that $r \models \Si$ and $r \not\models \phi$. This concludes the proof of Theorem \ref{partial1}.\qed
\end{proof}


\section{The Finite  Implication Problem}\label{sect: implication}

In Subsect. \ref{separation} we will show that the general and the finite implication do not coincide for keys and unary independence atoms. Using these results, we will show in Subsect. \ref{noaxioms} that for no $k$, there exists a $k$-ary axiomatization of the corresponding finite implication problem.


\subsection{Separation of the Finite and the  General Implication Problems}\label{separation}
For $n \geq 2$, let $R_n:= \{A_i,B_i : 1 \leq i \leq n\}$ be a relation schema, and let $\Si_n:= \{\indepc{A_i}{B_i} : 1 \leq i \leq n\} \cup \{\key{B_iA_{i+1}} : 1 \leq i \leq n, i $ modulo $ n\}$.\footnote{$\Si_n$ forms a smiley face of $n-1$ eyes. For instance, $\Si_7$ is illustrated in Figure \ref{kuva0} where each pair of attributes connected by an edge represents a key of $\Si_7$.} In this subsection we will show in Lemma \ref{finimpl} and \ref{genlem} that $\Si_n \models_{\rm FIN} \key{A_1B_1}$, for $n \geq 2$, and $\Si_2 \not\models \key{A_1B_1}$. Hence we will obtain the main result of this subsection.
\begin{figure}[h]
\center
\begin{tikzpicture}[xscale=.9,yscale=1.3 ]

\draw [-] [ ,  thick, shorten >=10pt, shorten <=10pt]  (1,0) to [out=80,in=110] (2,0);
\draw [-] [ ,  thick, shorten >=10pt, shorten <=10pt]  (3,0) to [out=80,in=110] (4,0);
\draw [-] [ ,  thick, shorten >=10pt, shorten <=10pt]  (5,0) to [out=80,in=110] (6,0); \draw [-] [ ,  thick, shorten >=10pt, shorten <=10pt]  (7,0) to [out=80,in=110] (8,0);
\draw [-] [ ,  thick, shorten >=10pt, shorten <=10pt]  (9,0) to [out=80,in=110] (10,0);
\draw [-] [ ,  thick, shorten >=10pt, shorten <=10pt]  (11,0) to [out=80,in=110] (12,0);
\draw [-] [ ,  thick, shorten >=12pt, shorten <=12pt]  (-.35,-.1) to [out=345,in=195] (13.35,-.1);

\node (a) at (0.5,0) {$\bot$};
\node (a) at (2.5,0) {$\bot$};
\node (a) at (4.5,0) {$\bot$};
\node (a) at (6.5,0) {$\bot$};
\node (a) at (8.5,0) {$\bot$};
\node (a) at (10.5,0) {$\bot$};
\node (a) at (12.5,0) {$\bot$};

\node (a) at (0.5,0) {$\bot$};
\draw (0,0)  node[minimum size=.5cm,draw] {$A_1$};
\draw (1,0)  node[minimum size=.5cm,draw] {$B_1$};
\draw (2,0)  node[minimum size=.5cm,draw] {$A_2$};
\draw (3,0)  node[minimum size=.5cm,draw] {$B_2$};
\draw (4,0)  node[minimum size=.5cm,draw] {$A_3$};
\draw (5,0)  node[minimum size=.5cm,draw] {$B_3$};
\draw (6,0)  node[minimum size=.5cm,draw] {$A_4$};
\draw (7,0)  node[minimum size=.5cm,draw] {$B_4$};
\draw (8,0)  node[minimum size=.5cm,draw] {$A_5$};
\draw (9,0)  node[minimum size=.5cm,draw] {$B_5$};
\draw (10,0)  node[minimum size=.5cm,draw] {$A_6$};
\draw (11,0)  node[minimum size=.5cm,draw] {$B_6$};
\draw (12,0)  node[minimum size=.5cm,draw] {$A_7$};
\draw (13,0)  node[minimum size=.5cm,draw] {$B_7$};

\end{tikzpicture}
\caption{$\Si_7$} \label{kuva0}
\end{figure}

\begin{lemma}\label{finimpl}
For $n \geq 2$, $\Si_n \models_{\rm FIN} \key{A_1B_1}$.
\end{lemma}
\begin{proof}
Let $n\geq 2$, and let $r$ be a finite relation over $R_n$ such that $r  \models \Si_n$. We show that $r \models \key{A_1B_1}$. First note that since $r \models \key{B_nA_1}$, we obtain that
\begin{equation}\label{equ1}
|r|= |r(B_nA_1)| \leq |r(B_n)|\cdot |r(A_1)|.
\end{equation}
Let then $2 \leq i \leq n$, and assume that $|r(B_i)|= m$. Then since $r \models \indepc{A_i}{B_i}$, each member of $r(A_i)$ has at least $m$ repetitions in $r$, that is, $|r(A_i = b)| \geq m$ for each $b \in r(A_i)$. Since $r \models \key{B_{i-1}A_i}$, we hence obtain that $|r(B_{i-1})| \geq m$ when $|r(B_i)| \leq |r(B_{i-1})|$. Therefore we conclude that $|r(B_n)| \leq |r(B_1)|$ when $|r| \leq |r(B_1)|\cdot |r(A_1)|$ by \eqref{equ1}. But now since $r \models \indepc{A_1}{B_1}$, we obtain that $|r(B_1)|\cdot |r(A_1)| = |r(B_1A_1)|$ from which the claim follows. \qed

\end{proof}
The following lemma can be proved by constructing a counter example for $\Si_2 \models \key{A_1B_1}$, similar to the one presented in the proof of Theorem \ref{partial1}.

\begin{lemma}\label{genlem}
$\Si_2 \not\models \key{A_1B_1}$.
\end{lemma}

\begin{proof}
We will construct a countably infinite relation $r$ over $R_2$ witnessing $\Si\not\models \phi$. For this we will inductively define an increasing chain (with respect to $\sub$) of finite relations $r_n$ over $R_2$ such that $r_1 \not\models \key{A_1B_1}$ and, for $n \geq 1$,
\begin{enumerate}
\item $r_n \models$ $\begin{cases} \key{B_2A_1},\\
   \key{B_1A_2},\end{cases}$
\item $r_n \models$ $\begin{cases}
\indepc{A_1}{B_1} \textrm{ if $n$ is odd},\\
 \indepc{A_2}{B_2}\textrm{ if $n$ is even}.
\end{cases}$
\end{enumerate}
Then, letting $r:= \bigcup_{n \geq 1} r_n$, we obtain that $r \models \Si$ and $r \not\models \phi$. The construction of relations $r_n$ is done as follows:

\begin{itemize}
\item For $n=1$, we let $r_1(A_1B_1A_2B_2):=\{(0,0,1,2),(0,0,3,4)\}$. Then $r_1 \models \key{B_2A_1}$, $r_1 \models  \key{B_1A_2}$ and $r_1 \models \indepc{A_1}{B_1}$.

\item Assume that  $r_n(A_1B_1A_2B_2)$ is a finite relation satisfying the induction assumption; we will construct a finite relation $r_{n+1}$ also satisfying the induction assumption. Without loss of generality we may assume that $n+1$ is even. Let $m$ be the maximal number occurring in $r_n$, and let $(a_1,b_1), \ldots ,(a_k,b_k)$ enumerate the set $(r_n(A_2) \times r_n(B_2)) \setminus r_n(A_2B_2)$. Note that this set is non-empty because otherwise, by the induction assumption, we would obtain a finite relation $r$ witnessing $\Si_2\not\models \key{A_1B_1}$, contrary to Lemma \ref{finimpl}. We then let
$$r_{n+1}(A_1B_1A_2B_2):= r_n \cup \{(a_i,b_i, m +2i-1,m +2i) : 1 \leq i \leq k\}.$$
By the construction and the induction assumption, it is straightforward to check that items 1 and 2 hold. This concludes the construction and the proof.
\end{itemize}\qed
\end{proof}

Hence,  from Lemma \ref{finimpl} and \ref{genlem}, we directly obtain the following corollary.
\begin{corollary}
For keys and unary independence atoms taken together, the finite implication problem and the general implication problem do not coincide.
\end{corollary}

\subsection{Non-axiomatizability of the Finite Implication}\label{noaxioms}
In this subsection we will show that for no $k$ there exists a $k$-ary axiomatization of the finite implication problem for unary independence atoms and keys taken together. For this, we first define, for $n \geq 2$, an upward closure of $\Si_n$ with respect to keys as follows:
$$\Cu{\Si_n}:= \Si_n \cup \{\key{D}: C  \sub D \sub R_n, \key{C} \in \Si_n\}.$$
Then we will show that $\Cu{\Si_n}$ is closed under $2n-1$-ary finite implication. Hence, and since $\key{A_1B_1} \not\in \Cu{\Si_n}$, it follows that the rule

$$
\cfrac{\Si_n}{\key{A_1B_1}}
 $$
for finite relations, is irreducible. That is, we cannot hope to deduce $\key{A_1B_1}$ from $\Si_n$ with a set of sound $2n-1$-ary rules.

Next we will show that $\Cu{\Si_n}$ is closed under $2n-1$-ary finite implication. For this, since $\Cu{\Si_n}$ is the closure of $\Si_n$ under the unary rule $\ax{7}$,
it suffices to show that all $2n-1$-ary finite implications of $\Si_n$ are included in  $\Cu{\Si_n}$. Namely, we will prove the following theorem.

\begin{theorem}\label{closed}
Let $n \geq 2$, $\Si' := \Si_n \setminus\{\psi\}$ where $\psi \in \Si_n$, and let $\phi$ be a $R_n$-key or a unary $R_n$-independence atom such that $\Si' \models_{\rm FIN} \phi$. Then $\phi \in \Cu{\Si_n}$.
\end{theorem}
This will be done in  Lemma \ref{mainlemma}, \ref{apu2}, \ref{apu3} and \ref{apu4} where in each case we consider one of the four different scenarios.

In the first case, Lemma \ref{mainlemma}, $\psi$ and $\phi$ are both keys (without loss of generality $\psi = \key{B_nA_1}$). In the proof of the lemma, we assume that $\phi\not\in \Cu{\Si_n}$ and show that $\Si' \not\models_{\rm FIN} \phi$ by constructing a finite relation $r$ such that $r \models \Si'$ and $r\not\models \phi$. Note that then, as in the proof of Theorem \ref{finimpl}, $|r(B_i)| \leq |r(B_{i-1})|$ and $|r(A_{i-1})| \leq |r(A_{i})|$, for $2 \leq i \leq n$. Hence we must construct $r$ so that $|r(A_1)|, |r(A_2)|, \ldots $ is increasing and $|r(B_1)|, |r(B_2)|, \ldots $ is decreasing. Moreover, if  $\phi = \key{D}$ for some $D \sub R_n$, and $A_iB_i \sub D$ for some $2 \leq i \leq n$, then we must have $|r(B_i)|+1\leq |r(B_{i-1})|$. Since, from $r \not\models \key{A_iB_i}$ it follows that there are two distinct $t,t'\in r$ with $t(A_iB_i)=t'(A_iB_i)$. Because $r\models \indepc{A_i}{B_i}$, then $t(A_i)$ must have at least $|r(B_i)|+1$ occurrences in column $A_i$ of $r$. Hence, by $r \models \key{B_{i-1}A_{i}}$, $r(B_{i-1})$ is at least of size $|r(B_i)|+1$. Analogously, if $A_iB_i \sub D$ for some $1 \leq i \leq n-1$, then we obtain that $|r(A_{i})|+1\leq |r(A_{i+1})|$.

\begin{lemma}\label{mainlemma}
Let $n \geq 2$, $\Si' := \Si_n \setminus\{\psi\}$ where $\psi \in \Si_n$ is a key, and let $\phi$ be a $R_n$-key such  that $\Si' \models_{\rm FIN} \phi$. Then $\phi \in \Cu{\Si_n}$.
\end{lemma}

\begin{proof}
By symmetry, we may assume that $\psi=\key{B_nA_1}$.
Let us assume to the contrary that  $\phi \not\in \Cu{\Si_n}$ where $\phi = \key{D}$ for some $D \sub R_n$. We will show that $\Si' \not\models_{\rm FIN} \phi$ by constructing a finite relation $r$ over  $R_n$ such that $r \models \Si'$ and $r \not\models \phi$.
For the construction of $r$, we will first associate each $1\leq i \leq n$ with natural numbers $a_i$ and $b_i$. Later $r$ will be defined inductively so that $|r(A_i)|= a_i$ and $|r(B_i)|=b_i$.

For defining $a_i$ and $b_i$, first let $m$ be the number of indices $1 \leq i \leq n$ such that $A_iB_i \sub D$, and let $M:= (m+3)!$. 
For $1 \leq i \leq n$, we define $a_i,b_i \geq 2$ as follows:\footnote{The  following definition of $a_i,b_i$, for $1 \leq i \leq 7$,  is illustrated in Figure \ref{kuva1} in case $D := \{A_1B_1A_3B_3A_5B_5A_7\}$. Note that in the example $m=3$ and $M = 720$.} We let $a_1:=2$, and if $a_i$ is defined, then we let
\begin{equation*}
b_i:=\begin{cases} \frac{M}{a_i}&\textrm{if }A_iB_i \not\sub D,\\
\frac{M}{a_i+1}&\textrm{if }A_iB_i \sub D,
\end{cases}
\end{equation*}
and $a_{i+1}:= \frac{M}{b_i}$.
It is straightforward to check that with this definition, for all $1 \leq i \leq n$, $a_i,b_i \in \mathbb{N}\setminus\{0,1\}$ and
\begin{equation}\label{(ii)} M =\begin{cases}
b_i \cdot a_{i+1}&\textrm{ if }i \leq n-1,\\
a_i\cdot b_i&\textrm{ if }A_iB_i \not\sub D,\\
(a_i+1)\cdot b_i &\textrm{ if } A_iB_i \sub D.
\end{cases}
\end{equation}



\begin{figure}[h]
\center
\begin{center}
    \scalebox{1.1}{\begin{tabular}{  | c | c | c | c | c | c | c | c | c | c | c | c | c | c | }
        $\pmb{a_1}$  & $\pmb{b_1}$ & $a_2$ & $b_2$ & $\pmb{a_3}$ & $\pmb{b_3}$ & $a_4$ & $b_4$  & $\pmb{a_5}$ & $\pmb{b_5}$ & $a_{6}$ & $b_{6}$ & $\pmb{a_7}$ & $b_7$  \\ \hline
  $2$ & $240 $ & $3$  & $240$ & $3$ & $180$ & $4$ & $180$ &  $4$  & $144$ & $5$ & $144$ & $5$  &  $144$ \\ \hline

\end{tabular}}
\end{center}
\caption{} \label{kuva1}
\end{figure}
We are now ready to define $r$. First we define two tuples $t$ and $t'$ as follows:\footnote{In our example, $t$ and $t'$ are defined as in Figure \ref{kuva2}.}
\begin{itemize}
\item $t(A)=0$ for all $A \in R_n$,
\item $t'(A) = \begin{cases}0&\textrm{ if }A \in  D,\\
1&\textrm{ if } A \in R_n\setminus  D.
\end{cases}$

\end{itemize}

\begin{figure}[h]
\center
\begin{center}
    \scalebox{1.1}{\begin{tabular}{  c | c | c | c | c | c | c | c | c | c | c | c | c | c | c | }
      &  $\pmb{A_1}$  & $\pmb{B_1}$ & $A_2$ & $B_2$ & $\pmb{A_3}$ & $\pmb{B_3}$ & $A_4$ & $B_4$  & $\pmb{A_5}$ & $\pmb{B_5}$ & $A_{6}$ & $B_{6}$ & $\pmb{A_7}$ & $B_7$  \\ \hline
 $t$ & $0$ & $0 $ & $0$  & $0$ & $0$ & $0$ & $0$ & $0$ &  $0$  & $0$ & $0$ & $0$ & $0$  &  $0$ \\ \hline
 $t'$ & $0$ & $0 $ & $1$  & $1$ & $0$ & $0$ & $1$ & $1$ &  $0$  & $0$ & $1$ & $1$ & $0$  &  $1$ \\ \hline

\end{tabular}}
\end{center}
\caption{} \label{kuva2}
\end{figure}
Since $\{t,t'\}\not\models \key{D}$, it suffices to embed $\{t,t'\}$ to a finite relation $r$ such that $r\models \Si'$. The relation $r$ will be defined inductively over columns. Namely, for $1 \leq i \leq n$, we will define a relation $r_i=\{t_0, \ldots ,t_{M-1}\}$ over  $R_i$ so that, for $i > 1$,
\begin{enumerate}
\item $r_i \models \Si' \upharpoonright R_i$,
\item $r_i(B_i)=\{0, \ldots ,b_i-1\}$ and $|r_i(B_i = l)|= \frac{M}{b_i}$, for each $0 \leq l \leq b_i-1$,
\item $t\upharpoonright R_i\setminus{A_1} = t_0$ and $t'\upharpoonright R_i\setminus{A_1} = t_1$.
\end{enumerate}
For $i=1$, we will introduce one new extra symbol $*$ that appears in column $A_1$. In the end of the construction, we let $r$ be obtained from $\bigcup_{1\leq i \leq n} r_i$ by replacing, in column $A_1$, $*$ with $0$. Then we will obtain that $r\models  \Si'$ and $\{t,t'\}=\{t_0,t_1\}$.
\begin{itemize}
\item Assume first that $i=1$. If $A_1B_1 \not\sub D$, then we let $r_1=\{t_0, \ldots ,t_{M-1}\}$ be a relation where $t_0(A_1B_1), \ldots , t_{M-1}(A_1B_1)$ is an enumeration of $\{0,1\}\times \{0, \ldots ,b_1-1\}$ such that $t_0(A_1B_1)=t(A_1B_1)$ and $t_1(A_1B_1) = t'(A_1B_1)$. Assume then that $A_1B_1 \sub D$. Then we let $r_1 =\{t_0, \ldots ,t_{M-1}\}$ be a relation where $t_0(A_1B_1), \ldots , t_{M-1}(A_1B_1)$ is an enumeration of $\{0,1,*\} \times \{0, \ldots ,b_1-1\}$ where $t_0(A_1B_1)=00$ and $t_1(A_1B_1)=*\hspace{.2mm}0$.

Since $ \Si' \upharpoonright R_1 = \{\indepc{A_1}{B_1}\}$, it is straightforward to check that items 1-3 hold.\footnote{The construction of $r_i$ up to $i=2$  is illustrated in Figure \ref{kuva3} in our example.} 

\begin{figure}[h]
\center
\begin{center}
    \scalebox{1.1}{\begin{tabular}{  c | c | c | c | c | c | c | c | c | c | c | c | c | c | c | }
      &  $\pmb{A_1}$  & $\pmb{B_1}$ & $A_2$ & $B_2$ & $\pmb{A_3}$ & $\pmb{B_3}$ & $A_4$ & $B_4$  & $\pmb{A_5}$ & $\pmb{B_5}$ & $A_{6}$ & $B_{6}$ & $\pmb{A_7}$ & $\pmb{B_7}$  \\ \hline
 $t_0$ & $0$ & $0 $ & $0$  & $0$ & $0$ & $0$ & $0$ & $0$ &  $0$  & $0$ & $0$ & $0$ & $0$  &  $0$ \\ \hline
 $t_1$ & $*$ & $0 $ & $1$  & $1$ & $0$ & $0$ & $1$ & $1$ &  $0$  & $0$ & $1$ & $1$ & $0$  &  $0$ \\ \hline
$t_2$ & $1$ & $0$ & $2$ & $0$ \\ \cline{1-5}
$t_3$ & $0$ & $1$  & $0$& $1$ \\  \cline{1-5}
$t_4$ & $*$ & $1$  & $1$ & $0$ \\ \cline{1-5}
$t_5$ & $1$ & $1$  & $2$ & $1$ \\ \cline{1-5}
$t_6$ & $0$ & $2$  & $0$& $2$ \\  \cline{1-5}
$t_7$ & $*$ & $2$  & $1$ & $2$ \\ \cline{1-5}
$t_8$ & $1$ & $2$  & $2$ & $2$ \\ \cline{1-5}
 $\vdots$ & $\vdots$ & $\vdots$ & $\vdots $  & $\vdots $\\ \cline{1-5}
$t_{717}$ & $0$ & $239$  & $0$ & $239$\\ \cline{1-5}
$t_{718}$ & $*$ & $239$  & $1$ & $239$\\ \cline{1-5}
$t_{719}$ & $1$ & $239$  & $2$ & $239$ \\ \cline{1-5}

\end{tabular}}
\end{center}
\caption{} \label{kuva3}
\end{figure}

\item Assume that $1\leq i < n $ and  $r_i(R_i) =\{t_0, \ldots ,t_{M-1}\}$ satisfies items 1-3. We will first extend $r_i$ to a relation $r^*(R_iA_{i+1})$ of size $M$ satisfying $\key{B_iA_{i+1}}$. First note that by the assumption $\key{D} \not\in \Cu{\Si_n}$, $B_iA_{i+1} \not\sub D$ when
\begin{equation}\label{erit}
t(B_{i}A_{i+1})\neq t'(B_{i}A_{i+1}).
\end{equation} Also by $\eqref{(ii)}$,  $M = b_i \cdot a_{i+1}$ when by item 2 of the induction assumption, $r_i(B_i)=\{0,\ldots ,b_i-1\}$ and $|r_i(B_i = l)|= a_{i+1}$, for each $0 \leq l \leq b_i-1$.
Hence and by \eqref{erit} we can define $r^*$ as a relation obtained from $r_i$ by extending each $t\in r_i$ with a value $t(A_{i+1}) \in \{0, \ldots ,a_{i+1}-1\}$ where $r^*(B_iA_{i+1})$ is an enumeration of $\{0, \ldots ,b_i-1\}\times \{0, \ldots ,a_{i+1}-1\}$ such that $t_0(B_iA_{i+1})=t(B_iA_{i+1})$ and $t_1(B_iA_{i+1})=t'(B_iA_{i+1})$. Since no repetitions occur in the enumeration, we obtain that $r^* \models \key{B_{i}A_{i+1}}$.

Next we will extend $r^*$ to $r_{i+1}$ satisfying items 1-3 of the induction claim. We have two cases:
\begin{itemize}
\item First assume that $A_{i+1}B_{i+1} \not\sub D$ when
\begin{equation}\label{erit2}
t(A_{i+1}B_{i+1})\neq t'(A_{i+1}B_{i+1}).
\end{equation}
Also by the previous construction and since $b_i = b_{i+1}$ by \eqref{(ii)}, $r^*(A_{i+1})=\{0, \ldots ,a_{i+1}-1\}$ and $|r^*(A_{i+1}= l)|= b_{i+1}$ for $0 \leq l \leq a_{i+1}-1$. Hence and by \eqref{erit2}, we can define $r_{i+1}$ as a relation obtained from $r^*$ by extending each $t \in r^*$ with a value $t(B_{i+1}) \in \{0, \ldots ,b_{i+1}-1\}$ where $r_{i+1}(A_{i+1}B_{i+1})$ is an enumeration of $\{0, \ldots ,a_{i+1}-1\}\times \{0, \ldots ,b_{i+1}-1\}$ such that $t_0(A_{i+1}B_{i+1})=t(A_{i+1}B_{i+1})$ and  $t_1(A_{i+1}B_{i+1})=t'(A_{i+1}B_{i+1})$. By \eqref{(ii)} and the construction it is straightforward to check that $r_{i+1}$ satisfies items 1-3 of the induction claim.

\item Assume then that $A_{i+1}B_{i+1} \sub D$. Then
\begin{equation}\label{samat}
t(A_{i+1}B_{i+1})= 00 = t'(A_{i+1}B_{i+1})
\end{equation}
and by \eqref{(ii)},
\begin{equation}\label{tää}
M =b_i \cdot a_{i+1}= (a_{i+1}+1)\cdot b_{i+1}.
\end{equation}
Recall also that  by \eqref{samat} and the previous construction, $r^*=\{t_0, \ldots, t_{M-1}\}$ is such that $t_0(A_{i+1})=t_1(A_{i+1})=0$,
$r^*(A_{i+1})=\{0, \ldots ,a_{i+1}-1\}$ and $|r^*(A_{i+1}= l)|= b_i$ for $0 \leq l \leq a_{i+1}-1$. Hence, and since $b_{i+1} < b_i$ by \eqref{tää}, we can also enumerate $r^*(A_{i+1})$ by pairs $(k,l) \in \{0, \ldots ,a_{i+1}-1\}\times \{0, \ldots ,b_{i}-1\}$ such that
\begin{itemize}
\item $t_{(k,l)}(A_{i+1}) = k$ for all $(k,l) \in \{0, \ldots ,a_{i+1}-1\}\times \{0, \ldots ,b_{i}-1\}$,
\item $t_{(0,0)} = t_0$,
\item $t_{(0,b_{i+1})} = t_1$.
\end{itemize}
By \eqref{samat} $r_{i+1}$ should now  be defined so that $r_{i+1}(A_{i+1}B_{i+1})$ has repetitions in the first two rows. Therefore, unlike in the first case, we cannot define $r_{i+1}$ as the relation extending $r^*$ with the values of $B_{i+1}$  that are obtained directly from the binary enumeration presented above. Instead, we let $r_{i+1}$ be obtained from $r^*$ by extending each $t_{(k,l)} \in r^*$ with
\begin{equation*}
t_{(k,l)}(B_{i+1})= \begin{cases}  l &\textrm{if }0 \leq l \leq b_{i+1}-1,\\
 N-1 &\textrm{if $b_{i+1} \leq l \leq b_i-1$ and $(k,l)$ is the $N$th}\\
&\textrm{member of $\{0, \ldots ,a_{i+1}-1\}\times \{b_{i+1}, \ldots ,b_i-1\}$}\\
&\textrm{in lexicographic order.}
\end{cases}
\end{equation*}
Then we obtain that $t_0(B_{i+1}) = t_1(B_{i+1})=0$. Moreover by \eqref{tää},
$$b_{i+1} = a_{i+1}(b_i-b_{i+1}),$$
and therefore $\{0, \ldots ,a_{i+1}-1\}\times \{b_{i+1}, \ldots ,b_i-1\}$ is of size $b_{i+1}$. Hence by the definition of $r_{i+1}$, we obtain that $r_{i+1}(B_{i+1})=\{0, \ldots ,b_{i+1}-1\}$ and $$|r_{i+1}(B_{i+1} = l)|= a_{i+1}+1 = \frac{M}{b_{i+1}},$$ for each $0 \leq l \leq b_{i+1}-1$. Finally, since $r_{i+1}(A_iB_i) = \{0, \ldots ,a_{i+1}-1\}\times \{0, \ldots ,b_{i+1}-1\}$, we obtain that $r_i \models \indepc{A_{i+1}}{B_{i+1}}$ when $r_{i+1} \models \Si' \upharpoonright R_{i+1}$. Hence $r_{i+1}$ satisfies the induction claim. This concludes the case $A_{i+1}B_{i+1} \sub D$ and the construction.

\end{itemize}
\end{itemize}

We then let $r$ be obtained from $\bigcup_{1\leq i \leq n} r_i$ by replacing, in column $A_1$, $*$ with $0$. Clearly $r \models A_1 \boto B_1$ and $\{t,t'\}=\{t_0,t_1\}$. Hence we obtain that $r\models  \Si'$ and  $r \not\models \key{D}$. This concludes the proof of Lemma \ref{mainlemma}.\qed

\end{proof}

The remaining cases are stated in the following lemmata. In the next case $\psi$ is an independence atom and $\phi$ is a key.

\begin{lemma}\label{apu2}

Let $n \geq 2$, $\Si' := \Si_n \setminus\{\psi\}$ where $\psi \in \Si_n$ is a unary independence atom, and let $\phi$ be a $R_n$-key such that $\Si' \models_{\rm FIN} \phi$. Then $\phi \in \Cu{\Si_n}$.
\end{lemma}

\begin{proof}
By symmetry, we may assume that $\psi=\indepc{A_1}{B_1}$. Let us assume to the contrary that $\phi \not\in \Cu{\Si_n}$ where $\phi = \key{D}$ for some $D \sub R_n$. We will show that $\Si' \not\models_{\rm FIN} \phi$.
First we define $\Si^*:= \Si_n\setminus\{\key{B_nA_1}\}$. Then by the proof of Lemma \ref{mainlemma}, there exists a finite relation $r^*= \{t_0, \ldots , t_{M-1}\}$ such that $r^*\models \Si^*$, $\{t_0,t_1\}\not\models \key{D}$, $t_0(X) = 0$ for all $X\in R_n$, and \\
$$t_1(X) =\begin{cases}0 &\textrm{ if }X \in D,\\
1 &\textrm{ if }X \in R_n \setminus D.
\end{cases}$$
We let $r$ be obtained\footnote{See Fig. \ref{kuva3.5}} from $r^*$ by replacing, for $0 \leq i \leq M-1$, $t_i(A_1)$ with
\begin{itemize}
\item    $\hspace{.25cm}i\hspace{.4cm}\textrm{  if  } i\neq 1$,
\item  $\begin{cases} 0 &\textrm{ if }i = 1 \textrm{ and } B_n \not\in D,\\
1 &\textrm{ if } i=1 \textrm{ and } B_n\in D.
\end{cases}$
\end{itemize}
From the definition of $r$ and the fact that $A_1B_n \not\sub D$ it follows that $r \not\models \key{D}$ and $r \models \Si^* \setminus \{ \indepc{A_1}{B_1} \}$. For $r \models \Si'$, we still need to show that $r \models \key{B_nA_1}$. Because of the definition of $t_i(A_1)$ in $r$, $\key{B_nA_1}$ could be violated only in $\{t_0,t_1\}$. In that case we would have $t_1(A_1B_1)=00$ in $r$ which contradicts with the definitions. Hence we obtain that $r \not\models \key{D}$ which concludes the proof.\qed

\end{proof}
\begin{figure}[h]
\center
\begin{center}
    \scalebox{1.1}{\begin{tabular}{  c | c | c | c | c | c | c | c | c | c | c | c | c | c | c | }
      &  $A_1$  & $B_1$ & $A_2$ & $B_2$ & $A_3$ & $B_3$ &  $\ldots $ & $\ldots $ & $A_{n-1}$ & $B_{n-1}$ & $A_n$ & $B_n$  \\ \hline
 $t_0$ & $0$ & $ $ &   $$&  &  &  &  &  &  $$  &  $$ & $$  &  $0$ \\ \hline
 $t_1$ & $0$ & $ $ & $$  & $$ & $$ & $$ & $$ &  $$ & $$ & $$ & $$  &  $1$ \\ \hline
$t_2$ & $2$ & $$ & $$ &&&&&&&&& $y_2$\\ \cline{1-13}
$t_3$ & $3$ & $$  & $$& $$  &&&&&&&& $y_3$ \\  \cline{1-13}
$\vdots$ & $\vdots$ & $\vdots$ & $\vdots $&$\vdots $&$\vdots $&$\vdots $&$\vdots $&$\vdots $&$\vdots $&$\vdots $&$\vdots $& $\vdots $\\ \hline
$t_{M-2}$ & $M-2$ & $$  & $$ & $$ &&&&&&&&$y_{M-2}$\\ \cline{1-13}
$t_{M-1}$ & $M-1$ & $$  & $$ & $$ &&&&&&&&$y_{M-1}$\\ \cline{1-13}

\end{tabular}}
\end{center}
\caption{$r$ in case $B_n \not\in D$} \label{kuva3.5}
\end{figure}

In the third case $\psi$ is a key and $\phi$ is an independence atom.
\begin{lemma}\label{apu3}

Let $n \geq 2$, $\Si' := \Si_n \setminus\{\psi\}$ where $\psi \in \Si_n$ is a key, and let $\phi$ be a unary $R_n$-independence atom such that $\Si' \models_{\rm FIN} \phi$. Then $\phi \in \Cu{\Si_n}$.
\end{lemma}

\begin{proof}
By symmetry, we may assume that $\psi = \key{B_nA_1}$. Assume to the contrary that $\phi \not\in \Cu{\Si_n}$. We will show that $\Si' \not\models_{\rm FIN} \phi$. 
Due to $\ax2$
and by symmetry of $\Si'$, it suffices to consider only the cases where $\phi= \indepc{A_i}{Y}$, for some $1 \leq i \leq n$ and $Y\in R_n\setminus\{ B_i\}$.

So let $1 \leq i \leq n$. We will construct two finite relations $r$ and $r'$ such that
\begin{enumerate}
\item $r \models \Si'$ and $r'\models \Si'$,
\item $r\not\models$ $\begin{cases}\indepc{A_i}{A_j} &\textrm{for }j\leq i,\\
\indepc{A_i}{B_j} &\textrm{for } j> i,
\end{cases}$
\item  $r'\not\models$ $\begin{cases}\indepc{A_i}{A_j} &\textrm{for }j> i,\\
\indepc{A_i}{B_j} &\textrm{for } j< i.
\end{cases}$
\end{enumerate}
We let $r:= \{t_0,t_1,t_2,t_3\}$ where we define, for $X \in R_n$,
\begin{itemize}
\item $t_0(X)= 0$,
\item $t_1(X)=$$\begin{cases}0 &\textrm{if }X=A_j\textrm{ for } j \leq i,\textrm{ or }X= B_j\textrm{ for } j > i,\\
1 &\textrm{otherwise,}
\end{cases}$
\item $t_2(X)=$$\begin{cases}0 &\textrm{if }X=B_i,\\
1 &\textrm{otherwise,}
\end{cases}$
\item $t_3(X)=$$\begin{cases}0 &\textrm{if }X=B_j \textrm{ for } j < i,\textrm{ or }X=A_j \textrm{ for } j> i,\\
1 &\textrm{otherwise.}
\end{cases}$
\end{itemize}
\begin{figure}[h]
\center
\begin{center}
    \scalebox{1.1}{\begin{tabular}{  c | c | c | c | c | c | c | c | c | c | c | c | c | c | c | c | c |  }
      &  $A_1$  & $B_1$ &  $\ldots $ & $\ldots $ & $A_{i-1}$ & $B_{i-1} $ & $A_i$ & $B_i $ & $A_{i+1}$ & $B_{i+1}$ & $\ldots $ & $\ldots $ & $A_n$ & $B_n$  \\ \hline
 $t_0$ & $0$ & $ 0$ &   $0$& $0$ & $0$  & $0$ & $0$ & $0$ &  $0$  &  $0$ & $0$  &  $0$ &$0$ &$0$ \\ \hline
 $t_1$ & $0$ & $ 1$ & $0$  & $1$ & $0$ & $1$ & $0$ &  $1$ & $1$ & $0$ & $1$  &  $0$ &$1$ &$0$ \\ \hline
$t_2$ & $1$ & $1$ & $1$ & $1$ &$1$& $1$ &$1$&$0$&$1$&$1$&$1$& $1$ &$1$ &$1$\\ \hline
$t_3$ & $1$ & $0$  & $1$& $0$  &$1$&$0$&$1$&$1$&$0$&$1$&$0$& $1$ &$0$ &$1$\\  \hline

\end{tabular}}
\end{center}
\caption{$r$} \label{kuva4}
\end{figure}
Then we let $r':=\{t_0,t_4\}$ where we define, for $X \in R_n$,
\begin{itemize}
\item $t_4(X)=$$\begin{cases}0 &\textrm{if }X=B_j\textrm{ for } j < i,\textrm{ or }X= A_j \textrm{ for }j \geq i,\\
1 &\textrm{otherwise.}
\end{cases}$
\end{itemize}
\begin{figure}[h]
\center
\begin{center}
    \scalebox{1.1}{\begin{tabular}{  c | c | c | c | c | c | c | c | c | c | c | c | c | c | c | c | c |  }
      &  $A_1$  & $B_1$ &  $\ldots $ & $\ldots $ & $A_{i-1}$ & $B_{i-1} $ & $A_i$ & $B_i $ & $A_{i+1}$ & $B_{i+1}$ & $\ldots $ & $\ldots $ & $A_n$ & $B_n$  \\ \hline
 $t_0$ & $0$ & $ 0$ &   $0$& $0$ & $0$  & $0$ & $0$ & $0$ &  $0$  &  $0$ & $0$  &  $0$ &$0$ &$0$ \\ \hline
 $t_4$ & $0$ & $ 1$ & $0$  & $1$ & $0$ & $1$ & $1$ &  $0$ & $1$ & $0$ & $1$  &  $0$ &$1$ &$0$ \\ \hline

\end{tabular}}
\end{center}
\caption{$r'$} \label{kuva5}
\end{figure}
It is straightforward to check that items 1-3 hold. This concludes the proof of Lemma \ref{apu3}.\qed

\end{proof}

In the last case both $\psi$ and $\phi$ are independence atoms.

\begin{lemma}\label{apu4}
Let $n \geq 2$, $\Si' := \Si_n \setminus\{\psi\}$ where $\psi \in \Si_n$ is a unary independence atom, and let $\phi$ be a unary $R_n$-independence atom such that $\Si' \models_{\rm FIN} \phi$. Then $\phi \in \Cu{\Si_n}$.
\end{lemma}

\begin{proof}
By symmetry, we may assume that $\psi= \indepc{A_1}{B_1}$.
Let us assume to the contrary that $\phi \not\in \Cu{\Si_n}$. We will show that $\Si' \not\models_{\rm FIN} \phi$. Analogously to the proof of Lemma \ref{apu3}, it suffices to consider only the cases where $\phi= \indepc{A_i}{Y}$, for some $1 \leq i \leq n$ and $Y\in R_n\setminus\{ B_i\}$. Let $1 \leq i \leq n$. We will construct four relations $r_0,r_1,r_2,r_3$ such that
\begin{enumerate}
\item $r_i \models \Si'$ for $i=0,1,2,3$,
\item  $r_0\not\models\indepc{A_i}{A_j} \textrm{ for }1 \leq j\leq  n$,
\item $r_1\not\models\indepc{A_1}{B_j} \textrm{ for }1< j$,
\end{enumerate}
and if $1 < i$,
\begin{itemize}
\item[4.]  $r_2\not\models\indepc{A_i}{B_j} \textrm{ for }j< i$,
\item[5.]  $r_3\not\models\indepc{A_i}{B_j} \textrm{ for }i< j$,

\end{itemize}

We let $r_0 :=\{t_0,t_1\}$ where we define, for $X \in R_n$,

\begin{itemize}
\item $t_0(X)= 0$,
\item $t_1(X)=$$\begin{cases}0 &\textrm{if }X=B_j \textrm{ for } j >1,\\
1 &\textrm{otherwise.}
\end{cases}$
\end{itemize}

\begin{figure}[h]
\center
\begin{center}
    \scalebox{1.1}{\begin{tabular}{  c | c | c | c | c | c | c | c | c | c | c | c | c | c | c | c | c |  }
      &  $A_1$  & $B_1$ &  $\ldots $ & $\ldots $ & $A_{i-1}$ & $B_{i-1} $ & $A_i$ & $B_i $ & $A_{i+1}$ & $B_{i+1}$ & $\ldots $ & $\ldots $ & $A_n$ & $B_n$  \\ \hline
 $t_0$ & $0$ & $ 0$ &   $0$& $0$ & $0$  & $0$ & $0$ & $0$ &  $0$  &  $0$ & $0$  &  $0$ &$0$ &$0$ \\ \hline
 $t_1$ & $1$ & $ 1$ & $1$  & $0$ & $1$ & $0$ & $1$ &  $0$ & $1$ & $0$ & $1$  &  $0$ &$1$ &$0$ \\ \hline

\end{tabular}}
\end{center}
\caption{$r_0$} \label{kuva6}
\end{figure}
Then we let $r_1:=\{t_0,t_2\}$ where we define, for $X \in R_n$,
\begin{itemize}
\item $t_2(X)=$$\begin{cases}0 &\textrm{if }X=A_j\textrm{ for } j >1,\\
1 &\textrm{otherwise.}
\end{cases}$
\end{itemize}
\begin{figure}[h]
\center
\begin{center}
    \scalebox{1.1}{\begin{tabular}{  c | c | c | c | c | c | c | c | c | c | c | c | c | c | c | c | c |  }
      &  $A_1$  & $B_1$ &  $\ldots $ & $\ldots $ & $A_{i-1}$ & $B_{i-1} $ & $A_i$ & $B_i $ & $A_{i+1}$ & $B_{i+1}$ & $\ldots $ & $\ldots $ & $A_n$ & $B_n$  \\ \hline
 $t_0$ & $0$ & $ 0$ &   $0$& $0$ & $0$  & $0$ & $0$ & $0$ &  $0$  &  $0$ & $0$  &  $0$ &$0$ &$0$ \\ \hline
 $t_2$ & $1$ & $ 1$ & $0$  & $1$ & $0$ & $1$ & $0$ &  $1$ & $0$ & $1$ & $0$  &  $1$ &$0$ &$1$ \\ \hline

\end{tabular}}
\end{center}
\caption{$r_1$} \label{kuva6}
\end{figure}
Assume then that $1<i$. We now let $r_2:= \{t_0,t_3\}$ where we define, for $X \in R_n$,
\begin{itemize}
\item $t_3(X)=$$\begin{cases}0 &\textrm{if }X=A_j \textrm{ for } 1<j <i,\textrm{ or }X=B_j \textrm{ for } i\leq j\leq n, \\
1 &\textrm{otherwise.}
\end{cases}$
\end{itemize}

\begin{figure}[h]
\center
\begin{center}
    \scalebox{1.1}{\begin{tabular}{  c | c | c | c | c | c | c | c | c | c | c | c | c | c | c | c | c |  }
      &  $A_1$  & $B_1$ &  $\ldots $ & $\ldots $ & $A_{i-1}$ & $B_{i-1} $ & $A_i$ & $B_i $ & $A_{i+1}$ & $B_{i+1}$ & $\ldots $ & $\ldots $ & $A_n$ & $B_n$  \\ \hline
 $t_0$ & $0$ & $ 0$ &   $0$& $0$ & $0$  & $0$ & $0$ & $0$ &  $0$  &  $0$ & $0$  &  $0$ &$0$ &$0$ \\ \hline
 $t_3$ & $1$ & $ 1$ & $0$  & $1$ & $0$ & $1$ & $1$ &  $0$ & $1$ & $0$ & $1$  &  $0$ &$1$ &$0$ \\ \hline

\end{tabular}}
\end{center}
\caption{$r_2$} \label{kuva6}
\end{figure}
For item 5, note that since $i < j \leq n$, we have that $i < n$. We let $r_3:=\{t_0,t_4,t_5,t_6\}$ where we define, for $X \in R_n$,
\begin{itemize}

\item $t_4(X)=$$\begin{cases}0 &\textrm{if }X=A_1,\textrm{ or }X=B_j\textrm{ for } j \leq i, \\
1 &\textrm{otherwise,}
\end{cases}$

\item $t_5(X)=$$\begin{cases}0 &\textrm{if }X=A_j \textrm{ for }1<j \leq i,\textrm{ or }X=B_j \textrm{ for } i< j, \\
1 &\textrm{otherwise,}
\end{cases}$

\item $t_6(X)=$$\begin{cases}0 &\textrm{if }X=A_j \textrm{ for }i<j , \\
1 &\textrm{otherwise.}
\end{cases}$
\end{itemize}

\begin{figure}[h]
\center
\begin{center}
    \scalebox{1.1}{\begin{tabular}{  c | c | c | c | c | c | c | c | c | c | c | c | c | c | c | c | c |  }
      &  $A_1$  & $B_1$ &  $\ldots $ & $\ldots $ & $A_{i-1}$ & $B_{i-1} $ & $A_i$ & $B_i $ & $A_{i+1}$ & $B_{i+1}$ & $\ldots $ & $\ldots $ & $A_n$ & $B_n$  \\ \hline

$t_0$ & $0$ & $0$  & $0$& $0$  &$0$&$0$&$0$&$0$&$0$&$0$&$0$& $0$ &$0$ &$0$\\  \hline

$t_4$ & $0$ & $0$ & $1$ & $0$ &$1$& $0$ &$1$&$0$&$1$&$1$&$1$& $1$ &$1$ &$1$\\ \hline

 $t_5$ & $1$ & $1 $ & $0$  & $1$ & $0$ & $1$ & $0$ &  $1$ & $1$ & $0$ & $1$  &  $0$ &$1$ &$0$ \\ \hline

 $t_6$ & $1$ & $ 1$ &   $1$& $1$ & $1$  & $1$ & $1$ & $1$ &  $0$  &  $1$ & $0$  &  $1$ &$0$ &$1$ \\ \hline

\end{tabular}}
\end{center}
\caption{$r_3$} \label{kuva7}
\end{figure}
Again, it is straightforward to check that items 1-5 hold. This concludes the proof of Lemma \ref{apu4}.\qed

\end{proof}

From Lemma \ref{mainlemma}, \ref{apu2}, \ref{apu3} and \ref{apu4} we obtain Theorem \ref{closed}. Using this we can prove the following theorem.
\begin{theorem}
For no natural number $k$, there exists a sound and complete $k$-ary axiomatization of the finite implication problem for unary independence atoms and keys taken together.
\end{theorem}
\begin{proof}
Let $k$ be a natural number, and let $n$ be such that $2n > k$. Then $\Si_n \models_{\rm FIN} \key{A_1B_1}$ by Theorem \ref{finimpl}. However, by the unary rule $\ax{7}$ and Theorem \ref{closed}, the closure of $\Si_n$ under $k$-ary finite implication is $\Cu{\Si_n}$. Since $\key{A_1B_1} \not\in \Cu{\Si_n}$, the claim follows.\qed

\end{proof}
Note that due to $\ax{4}$ and $\ax{2}$, for any non-unary $R_n$-independence atom $X \boto Y$ there exists a unary $A \boto B \not\in \Cu{\Si_n}$ such that $\{X \boto Y\} \models A \boto B$. Hence Theorem \ref{closed} can be extended to the case where $\phi$ is an independence atom of any arity. Therefore we obtain the following corollary.
\begin{corollary}
For no natural number $k$, there exists a sound and complete $k$-ary axiomatization of the finite implication problem for independence atoms and keys taken together.
\end{corollary}

\section{Conclusion}
We have studied the implication problem of unary independence atoms and keys taken together, both in the general and in the finite case. We gave a finite axiomatization of the general implication problem and showed that the finite implication problem has no finite axiomatization. The non-axiomatizability result holds also in case the arity of independence atoms is not restricted to one. It remains open whether the general implication problem for arbitrary independence atoms and keys enjoys a finite axiomatization, and whether the finite implication problem is undecidable.

\bibliographystyle{splncs03}
\bibliography{biblio}

\begin{thebibliography}{10}
\providecommand{\url}[1]{\texttt{#1}}
\providecommand{\urlprefix}{URL }

\bibitem{DBLP:journals/jacm/BojanczykMSS09}
Bojanczyk, M., Muscholl, A., Schwentick, T., Segoufin, L.: Two-variable logic
  on data trees and {XML} reasoning. J. ACM  56(3) (2009)

\bibitem{DBLP:journals/jcss/CasanovaFP84}
Casanova, M.A., Fagin, R., Papadimitriou, C.H.: Inclusion dependencies and
  their interaction with functional dependencies. J. Comput. Syst. Sci.  28(1),
   29--59 (1984)

\bibitem{chandra85}
Chandra, A.K., Vardi, M.Y.: The implication problem for functional and
  inclusion dependencies is undecidable. SIAM Journal on Computing  14(3),
  671--677 (1985)

\bibitem{DBLP:journals/jacm/CosmadakisKV90}
Cosmadakis, S.S., Kanellakis, P.C., Vardi, M.Y.: Polynomial-time implication
  problems for unary inclusion dependencies. J. ACM  37(1),  15--46 (1990)

\bibitem{DBLP:journals/ipl/Demetrovics78}
Demetrovics, J.: On the number of candidate keys. Inf. Process. Lett.  7(6),
  266--269 (1978)

\bibitem{DBLP:journals/tcs/DemetrovicsKMST98}
Demetrovics, J., Katona, G.O.H., Mikl{\'o}s, D., Seleznjev, O., Thalheim, B.:
  Asymptotic properties of keys and functional dependencies in random
  databases. Theor. Comput. Sci.  190(2),  151--166 (1998)

\bibitem{DBLP:journals/tods/Fagin81}
Fagin, R.: A normal form for relational databases that is based on domains and
  keys. ACM Trans. Database Syst.  6(3),  387--415 (1981)

\bibitem{geiger:1991}
Geiger, D., Paz, A., Pearl, J.: Axioms and algorithms for inferences involving
  probabilistic independence. Information and Computation  91(1),  128--141
  (1991)

\bibitem{gradel10}
Gr\"adel, E., V\"a\"an\"anen, J.: Dependence and independence. Studia Logica
  101(2),  399--410 (2013), \url{http://dx.doi.org/10.1007/s11225-013-9479-2}

\bibitem{DBLP:conf/foiks/HannulaK14}
Hannula, M., Kontinen, J.: A finite axiomatization of conditional independence
  and inclusion dependencies. In: Beierle, C., Meghini, C. (eds.) Foundations
  of Information and Knowledge Systems - 8th International Symposium, FoIKS
  2014, Bordeaux, France, March 3-7, 2014. Proceedings. Lecture Notes in
  Computer Science, vol. 8367, pp. 211--229. Springer (2014)

\bibitem{herrmann:2006}
Herrmann, C.: On the undecidability of implications between embedded
  multivalued database dependencies. Information and Computation  122(2),  221
  -- 235 (1995)

\bibitem{parker:1980}
Jr., D.S.P., Parsaye-Ghomi, K.: Inferences involving embedded multivalued
  dependencies and transitive dependencies. In: Chen, P.P., Sprowls, R.C.
  (eds.) Proceedings of the 1980 ACM SIGMOD International Conference on
  Management of Data, Santa Monica, California, May 14-16, 1980. pp. 52--57.
  ACM Press (1980)

\bibitem{DBLP:conf/wollic/KontinenLV13}
Kontinen, J., Link, S., V{\"a}{\"a}n{\"a}nen, J.A.: Independence in database
  relations. In: Libkin, L., Kohlenbach, U., de~Queiroz, R.J.G.B. (eds.)
  WoLLIC. Lecture Notes in Computer Science, vol. 8071, pp. 179--193. Springer
  (2013)

\bibitem{DBLP:journals/jcss/LucchesiO78}
Lucchesi, C.L., Osborn, S.L.: Candidate keys for relations. J. Comput. Syst.
  Sci.  17(2),  270--279 (1978)

\bibitem{mitchell83}
Mitchell, J.C.: The implication problem for functional and inclusion
  dependencies. Information and Control  56(3),  154--173 (1983)

\bibitem{DBLP:conf/icdt/NiewerthS11}
Niewerth, M., Schwentick, T.: Two-variable logic and key constraints on data
  words. In: ICDT. pp. 138--149 (2011)

\bibitem{paredaens:1980}
Paredaens, J.: The interaction of integrity constraints in an information
  system. J. Comput. Syst. Sci.  20(3),  310--329 (1980)

\bibitem{thalheim:1991}
Thalheim, B.: Dependencies in relational databases. Teubner (1991)

\bibitem{vaananen07}
V\"a\"an\"anen, J.: Dependence Logic. Cambridge University Press (2007)

\bibitem{DBLP:journals/is/Wijsen09}
Wijsen, J.: On the consistent rewriting of conjunctive queries under primary
  key constraints. Inf. Syst.  34(7),  578--601 (2009)

\end{thebibliography}







\end{document}